\documentclass{article}
\usepackage{spconf,amsmath,graphicx}
\usepackage{amsmath, amsfonts, amssymb, amsbsy}
\usepackage{algorithm}
\usepackage{enumerate}
\usepackage{algorithmic}
\usepackage[sort,compress]{cite}
\usepackage{epsfig}
\usepackage{epstopdf}
\usepackage{amsthm}
\usepackage[titletoc]{appendix}

\newtheorem{theorem}{Theorem}

\newcommand{\beq}{\begin{equation}}
\newcommand{\eeq}{\end{equation}}

\newcommand{\lb}{\left(}
\newcommand{\rb}{\right)}

\newcommand{\dsum}{\displaystyle\sum}

\def\adots{\mathinner{\mskip0mu\raise0pt\vbox{\kern7pt\hbox{.}}\mskip3mu
          \raise4pt\hbox{.}\mskip3mu\raise8pt\hbox{.}\mskip0mu}}

\newfont{\bb}{msbm10 scaled 1100}

\newcommand{\tr}{\mbox{tr}}

\newcommand{\Nt}{N_{t}}
\newcommand{\Nr}{N_{r}}

\newcommand{\bmh}{\mbox{\boldmath $h$}}

\newcommand{\CN}{{\cal CN}}

\newcommand{\bmx}{{\mathbf x}}
\newcommand{\bmy}{{\mathbf y}}
\newcommand{\bmv}{{\mathbf v}}

\newcommand{\bmG}{{\mathbf G}}
\newcommand{\bmH}{{\mathbf H}}
\newcommand{\bmQ}{{\mathbf Q}}
\newcommand{\bmR}{{\mathbf R}}

\renewcommand{\bmh}{{\mathbf h}}

\newcommand{\bmA}{{\mathbf A}}

\newcommand{\bmD}{{\mathbf D}}

\newcommand{\bmX}{{\mathbf X}}
\newcommand{\bmY}{{\mathbf Y}}

\newcommand{\bmC}{{\mathbf C}}

\newcommand{\bmHt}{\widetilde{\bmH}}
\newcommand{\bmI}{{\mathbf I}}

\renewcommand{\dagger}{H}

\newcommand{\kbar}{\overline{k}}

\renewcommand{\Nt}{M}
\renewcommand{\Nr}{N}

\newcommand{\matrice}[1]{\mbox{\bf #1}}

\newcommand{\mE}{\matrice{E}}

\newcommand{\bmHb}{\overline{\bmH}}

\title{A Refined Analysis of the Gap between Expected Rate for Partial CSIT and the Massive MIMO Rate Limit}
\name{Kalyana Gopala , Dirk Slock}
\address{EURECOM, Sophia-Antipolis, France\\
Email: {gopala,slock}@eurecom.fr}
\begin{document}
\ninept
%
\maketitle
\begin{abstract}
Optimal BeamFormers (BFs) that maximize the Weighted Sum Rate (WSR) for a Multiple-Input Multiple-Output (MIMO) interference broadcast channel (IBC) remains an important research area. Under practical scenarios, the problem is compounded by the fact that only partial channel state information at the transmitter (CSIT) is available. 
Hence, a typical choice of the optimization metric  is the Expected Weighted Sum Rate (EWSR). However, the presence of the expectation operator makes the optimization a daunting task. On the other hand,  for the particular, but significant, special case of  massive MIMO (MaMIMO), the EWSR converges to 
Expected Signal covariance Expected Interference covariance based WSR (ESEI-WSR) and this metric is more amenable to optimization. Recently, \cite{ShaoKin2017} considered a multi-user Multiple-Input Single-Output (MISO) scenario and proposed approximating the EWSR by ESEI-WSR. They then derived a constant bound for this approximation. This paper performs a refined analysis of the gap between EWSR and  ESEI-WSR criteria for finite antenna dimensions. 
\end{abstract}
\begin{keywords}
Beamforming, partial CSIT, EWSR, ESEI-WSR, MaMIMO
\end{keywords}
\section{Introduction}
\label{sec:intro}

Interference is the main limiting factor in wireless transmission.
Base stations (BSs) with multiple antennas are able to serve multiple Mobile Terminals (MTs) simultaneously, which is called
Spatial Division Multiple Access (SDMA) or 
Multi-User (MU) MIMO. We are particularly concerned here
with maximum Weighted Sum Rate (WSR) designs accounting for finite SNR. 
Typical approaches for maximizing WSR are
based on a link to Weighted Sum MSE (WSMSE) \cite{ChristensenTWC2008}  or an approach based on Difference of Convex function programming \cite{KimGiannaOptRAM} (which is actually better interpreted as an instance of majorization). However, these approaches rely on perfect channel CSIT, which is not practical. Hence, an alternative approach 
is to  maximize the EWSR for the case of partial CSIT. 

Partial CSIT formulations can typically be categorized as
either bounded error / worst case (relevant for quantization
error in digital feedback) or Gaussian error (relevant for analog
feedback, prediction error, second-order statistics information
etc.). The Gaussian CSIT formulation with mean and covariance
information was first introduced for SDMA
 (a Direction
of Arrival (DoA) based historical precedent of MU MIMO), in
which the channel outer product was typically replaced by the
transmit side channel correlation matrix, and worked out in
more detail for single user (SU) MIMO, e.g. \cite{FranciscoDirkAsil2005}. The use
of covariance CSIT was made in the context
of Massive MIMO \cite{YinGesbertJSAC2013}, where a not so rich propagation
environment leads to subspaces (slow CSIT) for the channel
vectors so that the fast CSIT can be reduced to the smaller
dimension of the subspace. Such CSIT (feedback) reduction
is especially crucial for Massive MIMO. Due to the difficulty in directly optimizing the EWSR metric, optimization of the expected WSMSE (EWSMSE), which is a lower bound for the EWSR,  was  proposed  in \cite{NegroDirkISWCS2012}. 
In fact, exact expressions exist for a number of MISO \cite{BjornsonTSP2010} and MIMO cases \cite{AlfanoVerduIntZurichSeminar2006}.
However, those expressions are very hard to interpret  and to optimize with respect to BFs.
This issue has led to the development of large system analysis to try to get simpler
expressions for the expected rate \cite{DumontInfoTheoryAsymptCovMIMO},\cite{TariccoInfoTheoryAsympMutualMIMO}. Recently, though under a single user MIMO setting, the authors \cite{KalySlockIcassp2017} used a large system approximation for the optimization of the EWSR metric under partial CSIT to counter the impact of Doppler created Inter Carrier Interference (ICI). On the other hand,  for the particular, but significant, special case of MaMIMO where the number of transmit antennas is large compared to the number of receive antennas, the EWSR converges to 
ESEI-WSR and this metric is more amenable to optimization. In another recent publication, \cite{ShaoKin2017}
considered a multi-user Multiple-Input Single-Output (MISO) scenario and proposed approximating the EWSR by ESEI-WSR. They then derived a constant bound for this approximation. The approximate metric was then used for optimization of the EWSR. Inspired by this, we perform a refined analysis of the gap between EWSR and  ESEI-WSR criteria  for finite antenna dimensions to evaluate the usefulness of using the ESEI-WSR metric (that is more mathematically tractable) instead of the EWSR.

  The main goal of this paper is to show that the much simpler expressions obtained in the ESEI approximation (MaMIMO limit) in fact exhibit only a finite and even small gap
to the exact expected rate. Towards this end, we first show in section \ref{subsec:gap_monotone} for a general non-zero mean correlated MIMO scenario  that the gap is monotonically increasing as a function of SNR and hence is maximum at infinite SNR. Then, we go about deriving this gap at infinite SNR for specific scenarios like uncorrelated MISO (section \ref{subsec:miso_iid}), correlated MISO (section \ref{subsec:gap_miso_corr}) and uncorrelated MIMO(section \ref{subsec:gap_mimo_zm_uncorr}).  The swift reduction in the gap with increasing number of antennas is clearly seen for the MISO scenarios.
The second order Taylor  Series Expansion of EWSR for a general MIMO setting is also derived in section \ref{subsec:TaylExp} and observed to concur with the infinite SNR limits for the gap derived independently. 
Henceforth, the term gap would refer to the gap between ESEI-WSR and the EWSR. In the following text, the notation $|\bmA|$ refers to the determinant of the matrix $\bmA$. $\CN(\mu,C)$ refers to a complex Gaussian distribution with mean $\mu$ and covariance $C$. In this paper, Tx may denote transmit/transmitter/transmission and Rx may denote receive/receiver/reception. 
\vspace{-5mm}
\section{MIMO IBC Signal Model}
\label{sec:mimo_ibc_sig_model}
\vspace{-3mm}
Consider  an IBC with $C$ cells with a total of $K$ users with $d_k$ streams per user. 
We shall consider a system-wide numbering of the users. User $k$ has $\Nr_k$ antennas is served by BS $b_k$.
The $\Nr_k\times d_k$ received signal at user $k$ in cell $b_k$ is,
\beq
\bmy_k\! =\! \underbrace{\bmH_{k,b_k}\, \bmG_k\, \bmx_k}_{\mbox{signal}} + \!\!\!\underbrace{\sum_{\stackrel{i\neq k}{b_i=b_k}} \!\!\bmH_{k,b_k}\,\bmG_i\, \bmx_i}_{\mbox{intracell interf.}} +
\!\!\underbrace{\sum_{j\neq b_k} \sum_{i: b_i=j} \! \bmH_{k,j}\,\bmG_i\, \bmx_i}_{\mbox{intercell interf.}}\! + \!\bmv_k
\label{eqIBCM1}
\eeq
where $\bmx_k$ is the intended (white, unit variance)  signal, 
$\bmH_{k,b_k}$ is the $\Nr_k\times\Nt_{b_k}$ channel from BS $b_k$ to user $k$. 
BS $b_k$ serves $K_{b_k}=\sum_{i:b_i=b_k}1$ users.
We consider a noise whitened signal representation so  that we get for the noise
$\bmv_k\sim\CN(0, I_{\Nr_k})$.
The $\Nt_{b_k}\times d_k$ spatial Tx filter or beamformer (BF) is $\bmG_k$.

The scenario of interest is that of partial CSIT available globally with all the BSs. The Gaussian CSIT model for the partial CSIT is 
\beq
\bmH_{k,b_k} = {\bmHb}_{k,b_k} +  {\bmHt}_{k,b_k}\,\bmC_t^{1/2}
\label{CSIT4}
\eeq
where ${\bmHb}_{k,b_k} = \mE \bmH_{k,b_k}$, and $\bmC_t^{1/2}$  is the Hermitian square-roots
of the  Tx side covariance matrices. The elements of  $\bmHt_{k,b_k}$ are $\mbox{ i.i.d. }\sim\CN(0,1)$.
\beq
\begin{array}{l}
\mE_{\bmH_{k,b_k}|\bmHb_{k,b_k}}(\bmH_{k,b_k}-{\bmHb}_{k,b_k})(\bmH_{k,b_k}-{\bmHb}_{k,b_k})^{H} = \tr\{\bmC_t\}\bmI_{N_k} \\
\mE_{\bmH_{k,b_k}|\bmHb_{k,b_k}}(\bmH_{k,b_k}-{\bmHb}_{k,b_k})^{H}(\bmH_{k,b_k}-{\bmHb}_{k,b_k}) = N_k\bmC_t 
\end{array}
\label{CSIT5}
\eeq
Note that the expectation is done over $\bmH_{k,b_k}$, for a known $\bmHb_{k,b_k}$. This is true for all the expectation operations done in this paper. However, as the parameter over which the expectation is done is clear from the context, henceforth, we just mention the expectation operator $\mE$ to reduce notational overhead. It is also of interest to consider the total Tx side correlation matrix,
\beq
 \mE \bmH_{k,b_k}\bmH^{H}_{k,b_k} = {\bmHb}_{k,b_k}{\bmHb}^{H}_{k,b_k} + \tr\{\bmC_t\}\bmI_{N_k}. 
\eeq

\subsection{Expected WSR (EWSR)}
Once the CSIT is imperfect, various optimization criteria could be considered, such as outage capacity. Here we shall consider the EWSR 
for a known channel mean ${\bmHb}$.
\beq
\begin{aligned}
\text{EWSR}(\bmG) & =  \mE \sum_k u_k \ln  | \mathbf{I} \!+\!\bmG_k^H\bmH_{k,b_k}^H\bmR_{k}^{-1}\bmH_{k,b_k} \bmG_k | \\
  &= \mE  \sum_{k=1}^K u_k  \;\lb \ln | \bmR_{k} | - \ln | \bmR_{\kbar} |  \rb.
  \end{aligned}
\label{eqWSR8}
\eeq
Here, $\bmG$ represents the collection of BFs $\bmG_k$,  $u_k$ are rate weights. 
\beq\mbox{}\!\!
\begin{array}{l}
\bmR_k =  \bmH_{k,b_k}\bmQ_k\,\bmH_{k,b_k}^H +   \bmR_{\kbar}\, ,\;  \bmQ_i = \bmG_i\bmG_i^H\, ,\\[1mm]
\bmR_{\kbar} =  \dsum_{i\neq k}\bmH_{k,b_i}\bmQ_i\,\bmH_{k,b_i}^H +   I_{\Nr_k}\, .
\end{array}
\label{eqWSR2}
\eeq
The EWSR cost function needs to be augmented with the power constraints
$\sum_{k:b_k=j}\tr\{\bmQ_k\} \leq P_j $ .

\subsection{MaMIMO limit and ESEI-WSR}

If the number of Tx antennas $M$ becomes very large, we get a convergence for any quadratic term of the form
\beq
\bmH\bmQ\bmH^H \stackrel{\Nt\rightarrow\infty}{\longrightarrow}
\mE \, \bmH\bmQ\bmH^H 
 = \bmHb\bmQ\bmHb^H \!\!+ \tr\{\bmQ\bmC_t\}\,\bmI\\
\eeq
\label{massive1}
and hence we get the following MaMIMO limit matrices
\beq
\begin{array}{l}
\breve{\bmR}_k = \breve{\bmR}_{\kbar} + \bmHb_{k,b_k} \bmQ_k\bmHb_{k,b_k}^H + \tr\{\bmQ_k\bmC_{t,k,b_k}\}\,\bmI_{N_k} \\
\breve{\bmR}_{\kbar} = I_{\Nr_k} + \dsum_{i\neq k}^K\lb\bmHb_{k,b_i} \bmQ_{i}\bmHb_{k,b_i}^H + \tr\{\bmQ_{i}\bmC_{t,k,b_i}\}\, \bmI_{N_k} \rb\\
\end{array}
\label{massive4}
\eeq
Now, typical approaches to solve the WSR (eg. the DC approach in  \cite{KimGiannaOptRAM} ) can be run to obtain the max EWSR BF.
We shall refer to this approach as the ESEI-WSR approach as (channel dependent) signal and interference covariance matrices are replaced
by their expected values. In the following sections, we analyze the gap between the EWSR and the ESEI-WSR to suggest an approximation of the first by the latter in the design of the BF.  We would like to remark here that the ESEI-WSR may also be interpreted as the WSR that would be obtained  if we assume that  the received signal and interference are also Gaussian.

\section{EWSR to ESEI-WSR gap Analysis}
\label{sec:gap_analysis}
We are interested in bounding the difference between  ESEI-WSR and the EWSR. 
%
At the level of each user $k$, we stack the channel estimates relevant for each user $k$.
\begin{equation}
\begin{aligned}
    {\mathbf{H}}_k &= [{\mathbf{H}}_{k,b_1} \cdots {\mathbf{H}}_{k,b_{k-1}}  \,\, {\mathbf{H}}_{k,b_k} \,\,  {\mathbf{H}}_{k,b_{k+1}}  \cdots {\mathbf{H}}_{k,b_K} ] \\
&= \bmHb_k +   \bmHt_k\mathbf{C}_{t,k}^{\frac{1}{2}}
\end{aligned}
\end{equation}
where  the elements of $\bmHt_k$ are i.i.d $\sim \mathcal{CN}(0,1)$ and $ \bmHb_k$ refers to the mean part of ${\mathbf{H}}_k$. $\mathbf{C}_{t,k}$ is a block diagonal matrix whose $i^{th}$ diagonal block is $\bmC_{t,k,b_i}$.
 Let ${\mathbf{Q}}$ be a block diagonal matrix with  each  diagonal block being ${\mathbf{Q}}_k$. ${\mathbf{Q}}_{\bar{k}}$ is similar to 
${\mathbf{Q}}$ but with the $k^{th}$ block diagonal set to all zeros. Then,
\begin{equation}
\begin{aligned}
 {{\mathbf{R}}}_{{k}} = \mathbf{I} +  {\mathbf{H}}_k {\mathbf{Q}}{\mathbf{H}}_k^{H} , \qquad
{{\mathbf{R}}}_{\bar{k}} = \mathbf{I} +  {\mathbf{H}}_k {\mathbf{Q}}_{\bar{k}} {\mathbf{H}}_k^{H}  
\end{aligned}
\end{equation}
\beq
\begin{aligned}
 \text{EWSR}(\bmG) &=   \sum_{k=1}^K u_k  \;\mE_{\bmH_k}  \lb \ln | \bmR_{k} | - \ln | \bmR_{\kbar} |  \rb \\
  = & \mE   \sum_{k=1}^K u_k  \;\lb \ln | \mathbf{I} +  {\mathbf{H}}_k {\mathbf{Q}}{\mathbf{H}}_k^{H} | - \ln | \mathbf{I} +  {\mathbf{H}}_k {\mathbf{Q}}_{\bar{k}} {\mathbf{H}}_k^{H} |  \rb 
  \end{aligned}
\label{eqWSR9}
\eeq
\beq
\begin{aligned}
& \text{ESEI-WSR}(\bmG)  \\
 &=  \sum_{k=1}^K u_k  \;\lb \ln | \mathbf{I} +  \mE \, {\mathbf{H}}_k {\mathbf{Q}}{\mathbf{H}}_k^{H} | - \ln | \mathbf{I} +  \mE \, {\mathbf{H}}_k {\mathbf{Q}}_{\bar{k}} {\mathbf{H}}_k^{H} |  \rb 
  \end{aligned}
\label{eqESEI-WSR9}
\eeq
Thus, the EWSR and ESEI-WSR have been rewritten in a convenient format so that one can focus on the gap between the two by comparing terms of the form $\mE  \ln | \mathbf{I} +  {\mathbf{H}}_k {\mathbf{Q}}{\mathbf{H}}_k^{H} | $ and  $\ln | \mathbf{I} +  \mE \, {\mathbf{H}}_k {\mathbf{Q}}{\mathbf{H}}_k^{H} |$. 
\subsection{Monotonicity of  gap  with SNR}
\label{subsec:gap_monotone}
For an SNR $\rho$, define
\beq
\Gamma(\rho) = \ln | \mathbf{I} +  \rho \mE  {\bmH}'_k {{\bmH}'_k}^{H} | - \mE  \ln | \mathbf{I} +  \rho  {\bmH}'_k {{\bmH}'_k}^{H} | 
\eeq
where ${\bmH}'_k \sim \CN(\bmHb'_k , \bmC)$, $ \bmHb'_k =    \frac{1}{\sqrt{\rho}}  {\bmHb}_k {\mathbf{Q}}^{\frac{1}{2}}$, and $\bmC = \frac{1}{{\rho}} \bmC_t^{\frac{1}{2}} \bmQ \bmC_t^{\frac{1}{2}} $.
Then,  $\mathbf{I} +   {\mathbf{H}}_k {\mathbf{Q}}{\mathbf{H}}_k^{H} = \mathbf{I} +   \rho \mE  {\bmH}'_k {{\bmH}'_k}^{H}$.

\begin{theorem}
\label{thm:monotone_inc}
$\Gamma(\rho)$  is monotonically increasing in $\rho$ 
\end{theorem}
\begin{proof}
By Jensen's inequality, $\Gamma_k \ge 0$.  To show the monotonicity, we show that the derivative with respect to $\rho$ is always non-negative. We omit the subscripts and superscripts on $\bmH$ for convenience. 
\beq
\begin{aligned} 
&\frac{\partial}{\partial \rho} \lb \ln |\mathbf{I} + \rho \mE \bmH \bmH^{H} | - \mE \ln| \mathbf{I} + \rho  \bmH \bmH^{H} |    \rb = \\
&\tr \lb \{\mathbf{I} + \rho \mE \bmH \bmH^{H}\}^{-1}\mE \bmH \bmH^{H}   -   \mE \lb \{\mathbf{I} + \rho  \bmH \bmH^{H}\}^{-1} \bmH \bmH^{H} \rb \rb
\end{aligned}
\eeq
Noting that, $\{\mathbf{I} + \rho \mE \bmH \bmH^{H}\}^{-1}\mE \bmH \bmH^{H} $ can be written as 
$\frac{1}{\rho}\bmI  - \frac{1}{\rho}\{\mathbf{I} + \rho \mE \bmH \bmH^{H}\}^{-1}$,
\beq
\begin{array} {l}
\frac{\partial}{\partial \rho} \lb \ln |\mathbf{I} + \rho \mE \bmH \bmH^{H} | - \mE \ln| \mathbf{I} + \rho  \bmH \bmH^{H} |    \rb = \\
\frac{1}{\rho} \tr \, \mE \lb \{\mathbf{I} + \rho  \bmH \bmH^{H}\}^{-1}\rb - \tr \, \frac{1}{\rho}\{\mathbf{I} + \rho \mE \bmH \bmH^{H}\}^{-1}    \ge 0
\end{array}
\eeq
where we have applied Jensen's inequality as  $\{\mathbf{I} + \rho  \bmH \bmH^{H}\}^{-1} $ is a convex function.
\end{proof}
As a result, the largest value of $\Gamma$ will be observed at infinite SNR for a general non-zero mean MIMO with channel $\bmH$ with arbitrary transmit covariance matrix. Now, following the same steps as in \cite{ShaoKin2017},
we can obtain, 
\beq
\begin{aligned}
\text{ESEI-WSR} -  \sum_{k=1}^K u_k \Gamma_{{k}}(\infty) & \le  
\text{ESEI-WSR} -  \sum_{k=1}^K u_k \Gamma_{{k}}(\rho) & \\ 
\le  \text{EWSR} \le   \\
 \text{ESEI-WSR} +  \sum_{k=1}^K u_k \Gamma_{\bar{k}}(\rho) & \le 
 \text{ESEI-WSR} +  \sum_{k=1}^K u_k \Gamma_{\bar{k}}(\infty). &
\end{aligned}
\eeq
In the above, $\Gamma_k$ and $\Gamma_{\bar{k}}$ are terms corresponding to the first and the second terms of equation \eqref{eqWSR9}. Remains now to obtain the $\Gamma(\infty)$ for different scenarios. However, we first look at the Taylor series expansion of EWSR to gain further insight.
\subsection{Second-Order Taylor Series Expansion of EWSR }
\label{subsec:TaylExp}
Consider the Taylor series expansion for matrices $\bmX$, $\bmY$ of dimension $N_k \times M$. 
\beq
\ln |\bmX+\bmY| \approx \ln |\bmX| + \tr{\bmX^{-1}\bmY} - \frac{1}{2} \tr{\bmX^{-1}\bmY\bmX^{-1}\bmY}
\eeq
Consider $\bmX+\bmY = \bmI + \rho \bmH\bmH^{H} $, $\bmH = \bmHb + \bmHt \bmC^{\frac{1}{2}}$, $ \bmHt \sim \CN(0,\bmI)$. For expansion around $\bmI + \rho \mE \, \bmH\bmH^{H}$, choose $\bmX = \bmI + \rho \mE \, \bmH\bmH^{H}$, $\bmY = \rho \lb \mE \, \bmH\bmH^{H} -  \bmH\bmH^{H} \rb $. Hence, we get,
\beq
\begin{aligned}
&\mE \, \ln | \bmI + \rho \bmH\bmH^{H} | \approx \ln | \bmI + \rho \mE \, \bmH\bmH^{H} | -\\
 &\frac{\rho^2}{2} \mE \, \tr \{\bmX^{-1}(\bmH\bmH^{H} - \mE \, \bmH\bmH^{H})  \bmX^{-1}(\bmH\bmH^{H} - \mE \, \bmH\bmH^{H}) \}
 \end{aligned}
\eeq
Using 4th order Gaussian moments \cite{JansStoicaTAC1988}, we get
\beq
\begin{aligned}
&\mE \, \ln | \bmI + \rho \bmH\bmH^{H} | \approx \ln | \bmI + \rho \mE \, \bmH\bmH^{H} | -  \frac{\rho^2}{2} \tr \bigg\{ \tr\{\bmX^{-1} \}^2 \bmC^2\\
&+2\tr\{\bmX^{-1} \} \bmHb^{H} \bmX^{-1} \bmHb\bmC - ( \bmHb^{H} \bmX^{-1} \bmHb )^2 \bigg\}.
 \end{aligned}
\eeq

Let us denote this second order approximation 
 by $\Gamma_2(\rho)$. i.e,
\beq
\begin{aligned}
\Gamma_2(\rho) &= \frac{\rho^2}{2} \tr \bigg\{ \tr\{\bmX^{-1} \}^2 \bmC^2  \\
&+2\tr\{\bmX^{-1} \} \bmHb^{H} \bmX^{-1} \bmHb\bmC - ( \bmHb^{H} \bmX^{-1} \bmHb )^2 \bigg\}.
\end{aligned}
\eeq

Consider the mean zero special case, $\bmHb = 0$. Then, $\mE \, \bmH\bmH^{H} =\tr\{\bmC\}\bmI$ and 
$\bmX=\bmI_{N_k} + \rho \tr\{\bmC\}\bmI_{N_k}$. Therefore,
\beq
\begin{aligned}
&\mE \, \ln | \bmI + \rho \bmH\bmH^{H} | \approx \ln (1 + \rho \tr(\bmC)) -\frac{\rho^2 N_k^2}{2}  \frac{ \tr \{ \bmC^2  \}}{(1+\rho \tr\{\bmC\})^2}.
 \end{aligned}
\eeq
At high SNR, as $\rho \to \infty$,
\beq
\label{eq:sec_order_taylor_infty}
\begin{aligned}
&\mE \, \ln | \bmI + \rho \bmH\bmH^{H} | \approx \ln (1 + \rho \tr(\bmC)) -\frac{N_k^2}{2}  \frac{ \tr \{ \bmC^2  \}}{(\tr\{\bmC\})^2}.
 \end{aligned}
\eeq
Thus,
\beq
\begin{aligned}
&\Gamma_2(\infty) =\frac{N_k^2}{2}  \frac{ \tr \{ \bmC^2  \}}{(\tr\{\bmC\})^2}
\end{aligned}
\eeq

Continuing from Theorem \ref{thm:monotone_inc}, we now determine the value of $\Gamma(\infty)$ for different scenarios.
\subsection{MISO independent and identically distributed (iid) channel}
\label{subsec:miso_iid}
In the MISO iid channel, the relevant metric is of the form  $\ln(1 + ||\bmh||^2)$, where $\bmh$ is the $ 1 \times M$ MISO channel vector. 
\begin{theorem}
\begin{equation}
\label{eq:exp_wsr_eqn_simplified}
0 \le  \ln(1+M\rho )  - \mathbf{E} \ln(1+\rho ||\bmh||^2 ) \le \gamma - \left( \sum_{k=1}^{M}\frac{1}{k}-\ln(M) \right) + \frac{1}{M} ,
\end{equation}
where $\rho$ is the SNR, $\gamma$ is Euler constant.
\end{theorem}

\begin{proof}
The proof is given in Appendix \ref{sec:coll_proofs}.
\end{proof}

Note that for $M=1$, the bound reduces to that in \cite{ShaoKin2017}, namely $\gamma$. Thus, this bound  is a much more refined and tighter bound than what is provided in \cite{ShaoKin2017}. 

We further explore the bound using the properties of the harmonic series. Define $\mathcal{H}_{M} = \sum_{k=1}^{M}\frac{1}{k} $. It is known that,
\beq
\label{eq:harm_ser}
\mathcal{H}_{M} = \ln (M) + \gamma + \frac{1}{2 M} - \frac{1}{12 M^2} + \frac{1}{120 M^4} \cdots
\eeq
Using this in \eqref{eq:exp_wsr_eqn_simplified}, we get
\beq
\gamma - (\mathcal{H}_{M} - \ln (M)) + \frac{1}{M} = \frac{1}{2 M} + \frac{1}{12 M^2} - \frac{1}{120 M^4} \cdots
\eeq
Thus, we see that the second order term for the bound is $\frac{1}{2 M}$, which is also in agreement with equation \eqref{eq:sec_order_taylor_infty}. In the iid case, $\bmC = \bmI_{M}$, hence,
  \beq
\frac{1}{2}  \frac{ \tr \{ \bmC^2  \}}{(\tr\{\bmC\})^2} =  \frac{\sum_{i=1}^{M} 1 }{ 2 (\sum_{i=1}^{M} 1)^2} = \frac{1}{2 M}
\eeq

\subsection{MISO correlated channel}
\label{subsec:gap_miso_corr}

\begin{theorem}
\begin{equation}
\label{eq:exp_wsr_eqn_simplified_corr_miso}
\begin{aligned}
0 & \le  \ln(1+\rho \sum_{i=1}^{p} \lambda_i )  - \mathbf{E} \ln(1+\rho || \bmh||^2)  \\
& \le \gamma - \lb \sum_{i=1}^{p} \frac{\ln \lambda_i}{ \pi_{l \ne i} (1 - \lambda_l/\lambda_i)} -\ln(\sum_{i=1}^{p} \lambda_i) \rb ,
\end{aligned}
\end{equation}
where $\rho$ is the SNR, $\gamma$ is Euler constant, $\lambda_i \cdots \lambda_p$ are the $p$ non-zero eigen values of the correlation matrix $\mE \, \bmh\bmh^{H}$. 
\end{theorem}

\begin{proof}
The proof is given in Appendix \ref{sec:coll_proofs}.
\end{proof}

From the second order Taylor series expansion (equation \eqref{eq:sec_order_taylor_infty}), the second order term of this bound is 
\beq
\frac{1}{2}  \frac{ \tr \{ \bmC^2  \}}{(\tr\{\bmC\})^2} =  \frac{\sum_{i=1}^{p} \lambda^2_i }{ 2 (\sum_{i=1}^{p} \lambda_i)^2} 
\eeq

\subsection{MIMO zero mean i.i.d channel}
\label{subsec:gap_mimo_zm_uncorr}

In a multi-user scenario, the regime of interest is $M \ge N_k$. To tackle this scenario, we 
first introduce the LDU (Lower Diagonal Upper triangular factorization) of the channel Gram matrix,
\beq
\bmH\bmH^{H} = \mathbf{L} \bmD \mathbf{L}^{H} =  (\mathbf{L} \bmD^{\frac{1}{2}}) (\mathbf{L} \bmD^{\frac{1}{2}})^{H}
\eeq
where $\mathbf{L}$ has unit diagonal and $\bmD$ is a diagonal matrix with diagonal entries ($\bmD_i$) greater than zero. The second factorization corresponds to a Cholesky decomposition. The Cholesky factorization of a Wishart matrix (such as 
$\bmH\bmH^{H}$) leads to, 
\beq
\begin{cases}
       \bmD_{i} \sim \frac{1}{2}\chi^2_{2(M-i+1)}, i \in 1 \cdots N_k \\
       \mathbf{L}_{i,j}\bmD^{\frac{1}{2}}_{i} \sim \CN(0,1) , i > j
\end{cases}
\eeq
which is also known as Bartlett's decomposition \cite{MultivarStats2008}. Note that $|\bmH\bmH^{H}| = |\mathbf{L} \bmD \mathbf{L}^{H}| = |\bmD|$. Hence, $\ln |\bmH\bmH^{H}| = \sum_{i=1}^{N_k} \ln |\bmD_i|$ and the MIMO case reduces to a sum of MISO scenarios, each having a $\chi^2$ distribution with a reducing number of degrees of freedom. Thus, reusing the results in section \ref{subsec:miso_iid}, we get
\beq
\begin{aligned}
  &\Gamma(\infty) =\\
  & \sum_{i=1}^{N_k} \lb  \gamma - \lb \sum_{k=1}^{M-i+1}  \frac{1}{k} - \ln(M-i+1) \rb + \frac{1}{M-i+1} \rb  \\
  &+ N_k \ln(M) -   \sum_{i=1}^{N_k} \ln(M-i+1) \\
  &= \sum_{i=1}^{N_k} \lb  \gamma - \lb \sum_{k=1}^{M-i}  \frac{1}{k} - \ln(M) \rb  \rb  
\end{aligned}   
\eeq
where the second term addresses the fact that the MISO gaps in section \ref{subsec:miso_iid} were computed with respect to the ESEI-WSR limit of $\ln(1+ \rho(M-i+1))$, whereas in the MIMO zero mean i.i.d scenario, the ESEI-WSR limit is 
$ N_k \ln(1+ \rho M)$. 
For illustration, let us also consider $M \gg N_k$. Then using the approximation of the Harmonic series, it can be easily shown that
$\Gamma(\infty) \approx \frac{N_k^2}{2 M}$,
which concurs with the second order Taylor series term in \eqref{eq:sec_order_taylor_infty}.


\section{Numerical Results}
\label{sec:num_res}
Figure \ref{fig:ewsr_1} verifies the infinite-SNR bounds for MISO i.i.d scenario by comparing them against the true values of the gap for different SNRs and different values of $M$. The true values of the gap are obtained by explicitly performing the integration in Matlab. As expected, the gap is zero at very low SNR. As the SNR increases, the gap monotonically increases to the infinite SNR limit, as predicted in section \ref{subsec:gap_monotone}. In addition, the gap reduces rapidly with increasing $M$. As the MIMO i.i.d case is a sum of MISO i.i.d scenarios, these curves apply to the MIMO i.i.d scenario as well.

\begin{figure}[htb]
\begin{minipage}[b]{1.0\linewidth}
  \centering
  \centerline{\includegraphics[width=\textwidth]{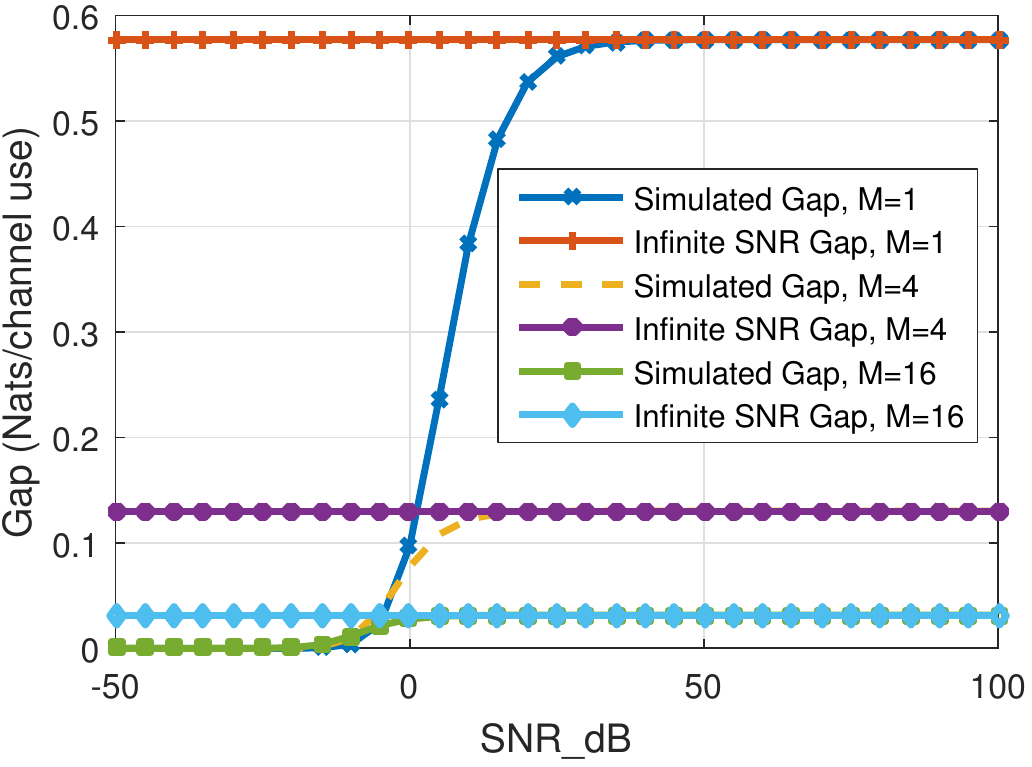}}
\caption{Gap between ESEI-WSR and EWSR for the MISO i.i.d scenario for different values of transmit antennas.}
\label{fig:ewsr_1}
\end{minipage}

\end{figure}
Further, to verify the goodness of the second order Taylor series approximation, Figure \ref{fig:ewsr_2} compares the true gap to the gap approximated  from the Taylor series expansion for a zero mean correlated MIMO scenario. This scenario is chosen as we expect gap to be maximum here. The number of receive antennas for each user was chosen as $N_k=4$. $\rho$ was chosen as 1000. As expected, the Taylor series approximation becomes more accurate with increasing number of Tx antennas. Indeed, even in this MIMO correlated scenario, the gap also reduces quickly as the number of Tx antennas increase.

\begin{figure}[htb]
\begin{minipage}[b]{1.0\linewidth}
  \centering
    \centerline{\includegraphics[width=\textwidth]{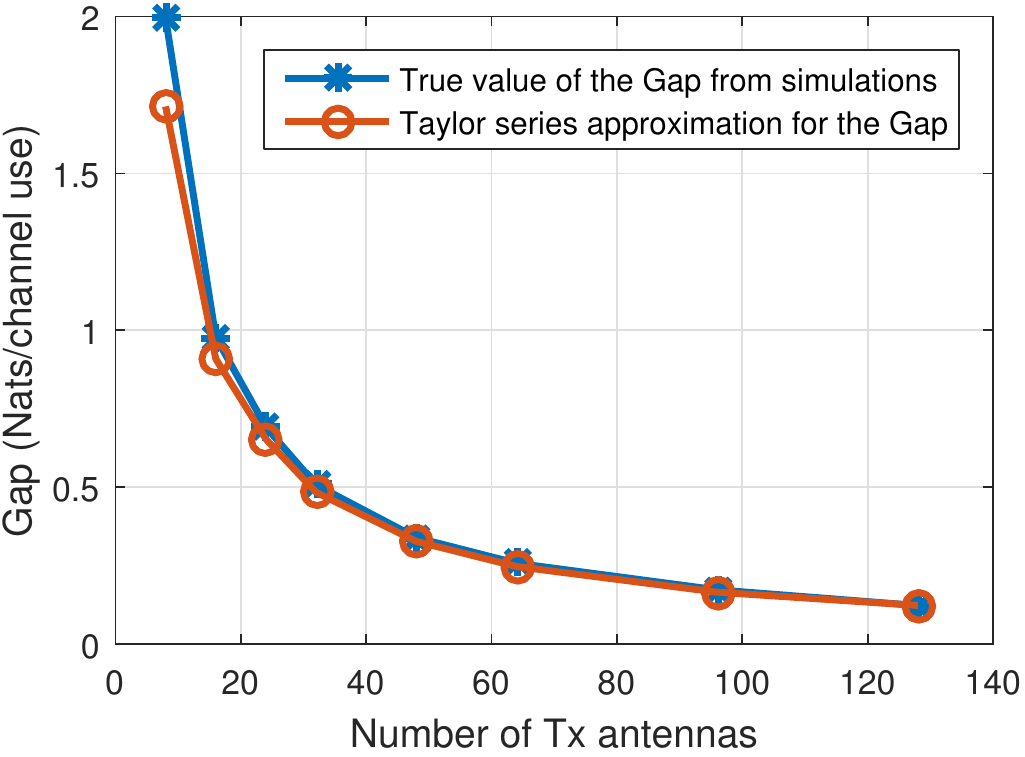}}
\caption{Comparison of the gap obtained from the second order Taylor series approximation and the true value of the gap for a MIMO correlated  scenario. The number of antennas at each receiver, $N_k$, is taken as 4.}
\label{fig:ewsr_2}
\end{minipage}

\end{figure}

\section{Conclusion}
\label{sec:conclusion}
In this paper, we have motivated the use of the ESEI-WSR metric (or the MaMIMO limit of the EWSR) for utility optimization involving partial CSIT. Towards this end, we presented a refined bound for the gap between EWSR and the ESEI-WSR. We first showed that the gap is maximum at infinite SNR. The results clearly show that the gap reduces with the number of transmit antennas - thereby concurring with the well known result for the MaMIMO limit.
The general case of correlated MIMO channel with non-zero mean is a future work to be addressed. However, we conjecture that in the case of a non-zero mean MIMO, the gap would further reduce based on the rice factor (the ratio of the power in the mean to that of the random part). 
However, a few comments are in order. Whenever $\Gamma(\infty)$ is closely approximated by $\Gamma_2(\infty)$
then $\Gamma(\rho)$ should be closely approximated by $\Gamma_2(\rho)$ also.
We can also observe that whenever the gap $\Gamma(\rho)$ gets small, the second-order term $\Gamma_2(\rho)$
becomes good, in the sense that
$\Gamma(\rho) = \Gamma_2(\rho) + O(\Gamma_2^2(\rho))$.

%


\section*{Acknowledgments}

EURECOM's research is partially supported by its industrial members:
ORANGE, BMW,  ST Microelectronics,
Sy\-man\-tec, SAP, Monaco Telecom, iABG,  by the projects HIGHTS (EU H2020)
and MASS-START (French FUI).

\begin{appendices}

\section{Collection of proofs}
\label{sec:coll_proofs}
\begin{proof} {\textit{of Theorem 2}}
To ease the notation, we take
$x =||\bmh||^2$, where $x$ is Chi-squared distributed with mean $M$.
For a $\chi^2$ distribution with mean $M$ and $2N_t$ degrees of freedom,
\begin{equation}
f_X(x) = \frac{x^{M-1} e^{-x}}{(M-1)!}. 
\end{equation}
As $\gamma = -\int_0^\infty e^{-x}  \ln(x) dx$,
at high SNR ($\rho \to \infty$), 
\beq
\begin{aligned}
\label{eq:exp_wsr_eqn_highsnr_mult_ant}
\mathbf{E}_{x  } \ln(1+\rho x) &= \int_0^\infty f_X(x)  \ln(\rho x) dx \\
&= \int_0^\infty \frac{x^{M-1} e^{-x}}{(M-1)!}  \ln( x) dx +  \ln(\rho).
\end{aligned}
\eeq
We note the following,
\begin{equation}
   \int  e^{-x}  \ln( x) dx = -e^{-x}  \ln( x)  + \text{Ei}(-x) \text{,} \text{Ei}(x) = -\int_{-x}^{\infty} \frac{e^{-t}}{t} dt 
 \eeq
 \beq   
 \begin{aligned}
  &  -\int_0^\infty \frac{x^{(M-2)}}{(M-2)!} \text{Ei}(-x)dx = \int_0^\infty \frac{x^{(M-2)}}{(M-2)!}  \int_{x}^{\infty} \frac{e^{-t}}{t} dt dx  \\
    &\qquad \qquad=  \int_0^\infty  \left( \int_{0}^{t}\frac{x^{(M-2)}}{(M-2)!} dx \right) \frac{e^{-t}}{t} dt \\
     &\qquad \qquad= \int_0^\infty \frac{t^{(M-1)}}{(M-1)!} \frac{e^{-t}}{t} dt 
     = \frac{1}{M-1}
\end{aligned}   
\end{equation}
Integrating by parts ($M \ge 2$),
\begin{equation}
\begin{aligned}
 &\int_0^\infty \frac{x^{M-1} e^{-x}}{(M-1)!}  \ln( x) dx =\\
 & \frac{x^{M-1}}{(M-1)!} \left(  -e^{-x}  \ln( x)  + \text{Ei}(-x)\right)_0^{\infty} -\\
 & \int_0^\infty \frac{x^{M-2}}{(M-2)!} \left(  -e^{-x}  \ln( x)  + \text{Ei}(-x)\right)
\end{aligned}
\end{equation}
The first part in the above equation is zero, so we only need to focus on the second portion of the integral. 
\begin{equation}
\begin{aligned}
& \int_0^\infty \frac{x^{M-1} e^{-x}}{(M-1)!}  \ln( x) dx \\
&= - \int_0^\infty \frac{x^{M-2}}{(M-2)!} \left(  -e^{-x}  \ln( x)  + \text{Ei}(-x)\right) \\
&= \int_0^\infty \frac{x^{M-2} e^{-x}}{(M-2)!} \ln( x) dx -  \int_0^\infty \frac{x^{M-2}}{(M-2)!}\text{Ei}(-x) \\
 &= \int_0^\infty \frac{x^{M-2} e^{-x}}{(M-2)!} \ln( x) dx +  \frac{1}{M-1}
\end{aligned}
\end{equation}
The above is a recursive equation, from where, we quickly deduce that,
\begin{equation}
\begin{aligned}
 \int_0^\infty \frac{x^{M-1} e^{-x}}{(M-1)!}  \ln( x) dx &= \int_0^\infty e^{-x} \ln( x) dx +  \sum_k^{M-1}\frac{1}{k} \\
 &= -\gamma + \sum_k^{M-1}\frac{1}{k} 
\end{aligned}
\end{equation}
Thus, we can now write \eqref{eq:exp_wsr_eqn_highsnr_mult_ant} as,
\begin{equation}
\label{eq:bound_miso_iid}
\begin{aligned}
\mathbf{E}_{x  } \ln(1+\rho x) &= \int_0^\infty \frac{x^{M-1} e^{-x}}{(M-1)!}  \ln( x) dx +  \ln(\rho) \\
&= -\gamma + \sum_{k=1}^{M-1}\frac{1}{k}  + \ln(\rho) \\
&=  -\gamma + \left( \sum_{k=1}^{M-1}\frac{1}{k}-\ln(M) \right)  + \ln(M \rho) 
\end{aligned}
\end{equation}
\end{proof}

\begin{proof}{\textit{of Theorem 3}.}

For a correlated MISO scenario, we can write equivalently,
\beq
\ln |1+ \rho || \bmh ||^2| = \ln |1+ \rho \sum_{i=1}^{p} \lambda_i | \bmh_i | ^2|, 
\eeq
where $\lambda_i, i \in 1 \cdots p$ are the non-zero eigen values of the correlation matrix $\mE \, \bmh\bmh^{H}$, scaled in such a manner that $\sum_{i=1}^{p} \lambda_i = M$. $ \bmh_i \sim \CN(0,1)$. We make the reasonable assumption that all the non-zero eigen values are unequal. In this case, the probability distribution is given \cite{HammarwallTSP2008} as $\sum_{i=1}^{p} \frac{e^{-\frac{x}{\lambda_i}}}{\lambda_i \pi_{l \ne i} (1 - \lambda_l/\lambda_i)}$, where $x = \sum_{i=1}^{p} \lambda_i | \bmh_i | ^2$. Thus,
at high SNR ($\rho \to \infty$), 
\beq
\begin{aligned}
\label{eq:uncorr_miso_wsr_eqn_highsnr_mult_ant}
&\mathbf{E}_{x  } \ln(1+\rho x) = \int_0^\infty \sum_{i=1}^{p} \frac{e^{-\frac{x}{\lambda_i}}}{\lambda_i \pi_{l \ne i} (1 - \lambda_l/\lambda_i)} \ln( x) dx +  \ln(\rho) \\
&= \sum_{i=1}^{p}   \frac{ \int_0^\infty \frac{1}{\lambda_i}  e^{-\frac{x}{\lambda_i}}\ln( x) dx}{\pi_{l \ne i} (1 - \lambda_l/\lambda_i)}  +  \ln(\rho) \\
&= \sum_{i=1}^{p} \frac{-\gamma + \ln \lambda_i}{ \pi_{l \ne i} (1 - \lambda_l/\lambda_i)}+  \ln(\rho) \\
&= -\gamma + \lb \sum_{i=1}^{p} \frac{\ln \lambda_i}{ \pi_{l \ne i} (1 - \lambda_l/\lambda_i)} -\ln(\sum_{i=1}^{p} \lambda_i) \rb +  \ln(\rho\sum_{i=1}^{p} \lambda_i).
\end{aligned}
\eeq

\end{proof}

\end{appendices}

%
%

\vfill\pagebreak

\bibliographystyle{IEEEbib}
\bibliography{strings,refs}

\end{document}